\documentclass[12pt]{article}
\usepackage{amssymb,amsmath,amsthm}

\begin{document}


\title{The non-existence of a $[[13,5,4]]$-quantum stabilizer code}
\author{J\"urgen Bierbrauer\\
Department of Mathematical Sciences\\
Michigan Technological University\\
Houghton, Michigan 49931 (USA)\\ \\
Stefano Marcugini and Fernanda Pambianco \\
Dipartimento di Matematica e Informatica\\
Universit\`a degli Studi di Perugia\\
Perugia (Italy)}

\maketitle
\newtheorem{Theorem}{Theorem}
\newtheorem{Proposition}{Proposition}
\newtheorem{Lemma}{Lemma}
\newtheorem{Definition}{Definition}
\newtheorem{Corollary}{Corollary}
\newtheorem{Example}{Example}
\def\nz{\mathbb{N}}
\def\gz{\mathbb{Z}}
\def\rz{\mathbb{R}}
\def\ef{\mathbb{F}}
\def\CC{\mathbb{C}}
\def\o{\omega}
\def\p{\overline{\omega}}
\def\e{\epsilon}
\def\a{\alpha}
\def\b{\beta}
\def\g{\gamma}
\def\d{\delta}
\def\l{\lambda}
\def\s{\sigma}
\def\bsl{\backslash}
\def\la{\longrightarrow}
\def\arr{\rightarrow}
\def\ov{\overline}
\def\sm{\setminus}
\newcommand{\D}{\displaystyle}
\newcommand{\T}{\textstyle}

\begin{abstract}
We solve one of the oldest problems in the theory of quantum stabilizer codes
by proving the non-existence of quantum $[[13,5,4]]$-codes.
\end{abstract}

\section{Introduction}
\label{introsection}

After the determination of the parameter spectrum of 
additive quantum codes of distance $3$ (see~\cite{d3quantspectrum}) the oldest open
existence problem for quantum stabilizer codes concerns the parameters $[[13,5,4]].$
We give a negative answer:

\begin{Theorem}
\label{maintheorem}
There is no $[[13,5,4]]$-quantum stabilizer code.
\end{Theorem}

The reduction of the problem of quantum error-correction to codes in symplectic geometry
essentially is in~\cite{CRSS}. For a geometric approach see also~\cite{quantgeom}.
We use the following definitions:

\begin{Definition}
\label{addcodebasicdef}
Let $k$ be such that $2k$ is a positive integer.
An additive quaternary $[n,k]_4$-code ${\cal C}$ (length $n,$
dimension $k$) is a $2k$-dimensional subspace of $\ef_2^{2n},$
where the coordinates come in pairs of two.
We view the codewords as $n$-tuples where the coordinate entries are elements
of $\ef_2^2.$
\par
A {\bf generator matrix} of ${\cal C}$ is a binary $(2k,2n)$-matrix whose
rows form a basis of the binary vector space ${\cal C}.$
\end{Definition}

In the case of quantum stabilizer codes we view the ambient space $\ef_2^{2n}$ as
a binary symplectic space, where each of the $n$ parameter sections corresponds to
a hyperbolic plane, equivalently a $2$-dimensional symplectic space. 
Each codeword is therefore a vector in the $2n$-dimensional
symplectic geometry over $\ef_2.$

\begin{Definition}
\label{quantumcodedef} 
A quaternary quantum stabilizer code is
an additive quaternary code $C$ which is contained in its dual, where duality
is with respect to the symplectic form.
\end{Definition}

Describe $C$ by a generator matrix $M.$ Each of the $n$ coordinate sections
contains $2$ columns which we view as points in binary projective space. The geometric
description of the quantum code is in terms of the system of $n$ lines
(the codelines) generated by those $n$ pairs of points. 
 
\begin{Definition}
\label{strengthdef}
Let $C$ be a quaternary additive code of length $n,$ with generator matrix $M.$
The {\bf strength} of $C$ is the largest number $t$ such that any $t$ codelines
are in general position.
\end{Definition}

Observe that the strength $t(C)$ is one less than the dual distance.

\begin{Definition}
\label{generalquantumcodedef}
An $[[n,m,d]]$-code $C$ where $m>0$ is a quaternary quantum stabilizer 
code of binary dimension $n-m$ satisfying the following: 
any codeword of $C^{\perp}$ having weight at most $d-1$ is in $C.$  
\par
The code is {\bf pure} if $C^{\perp}$ does not contain codewords of weight $\leq d-1,$
equivalently if $C$ has strength $t\geq d-1.$
\par
An $[[n,0,d]]$-code $C$ is a self-dual quaternary quantum stabilizer 
code of strength $t=d-1.$
\end{Definition}

The optimal parameters of quantum stabilizer codes of length $\leq 13$ are known,
with the sole exception of parameters $[[13,5,4]]$ (see the database in~\cite{Grassl}).
The remainder of the paper is dedicated to a proof of Theorem~\ref{maintheorem}.
Assume $C$ is a $[[13,5,4]]$-quantum code. In the next section we show that $C$ is necessarily pure.

\section{The purity of the code}
\label{puresection}

\begin{Proposition}
\label{pureprop}
Let $C$ be a $[[13,5,4]]$-quantum code. Then $C$ is pure.
\end{Proposition}

In general the geometric objects defined by the column pairs of a generator matrix
(which we called codelines) may be lines, points or even the empty set
(if the corresponding pair of columns has all entries $=0$).
The following basic fact follows from the definition:

\begin{Lemma}
\label{obviouslemma}
Whenever some $\leq d-1$ codelines of a quantum code of distance $d$
are not in general position
there is a hyperplane containing all the remaining codelines.
\end{Lemma}

In the remainder of this section we prove Proposition~\ref{pureprop}.
It follows from Proposition 3.1 of~\cite{quantgeom} that the codeobjects of $C$ are indeed
lines and that no line occurs more than once. Quantum code $C$ is therefore described by a set of $13$ different lines in $PG(7,2).$ Observe that the (quaternary) minimum weight of nonzero words in $C^{\perp}$ therefore is $\geq 2.$ 
As we are assuming that $C$ is not pure
there are three codelines $L_1,L_2,L_3$ contained in a subspace $PG(4,2).$

\begin{Lemma}
\label{supp3lemma}
Let $L_i,L_j,L_k$ be three codelines not in general position. Let
$v(\lbrace L_i,L_j,L_k\rbrace)\in C$ a nonzero codeword with support in coordinates
$i,j,k.$
\end{Lemma}

Observe that $v(\lbrace L_i,L_j,L_k\rbrace)$ in Lemma~\ref{supp3lemma} has weight $2$
or $3.$
\par
The $10$ remaining codelines are in a hyperplane $H.$
In the sequel we use basic facts concerning additive quaternary codes, see~\cite{08addcodes}. The nonexistence of a quaternary additive $[10,6.5,4]$ and its dual shows that the family of remaining codelines cannot have strength $3.$ It follows that
$L_4,L_5,L_6$ are in a subspace $PG(4,2).$ By Lemma~\ref{obviouslemma} there is a
hyperplane containing all codelines $\notin\lbrace L_4,L_5,L_6\rbrace .$
This shows that the $7$ codelines $\notin\lbrace L_1,\dots ,L_6\rbrace $ are contained in a secundum $S$ (a $PG(5,2)$). The non-existence of a quaternary $[7,4,4]$-code and its dual shows that three of the seven remaining lines ($L_7,L_8,L_9,$ say) are not in general position.
It follows from Lemma~\ref{obviouslemma} that $L_{10},\dots ,L_{13}$ are contained in
a subspace $PG(4,2).$
\par
We start from the information that some four lines which we now call $L_1,L_2,L_3,L_4$
are in a subspace $PG(4,2).$ The codewords
$v(\lbrace L_1,L_2,L_3\rbrace)$ and $v(\lbrace L_2,L_3,L_4\rbrace)$ show that there is a
secundum $S$ (a $PG(5,2)$) containing the remaining $9$ codelines. The usual argument,
based on the non-existence of a quaternary $[9,6,4]$-code, shows that there is a
$PG(4,2)$ containing $6$ codelines. 
\par
Start again and use the knowledge that some six codelines $L_1,\dots ,L_6$ are contained in a subspace $PG(4,2).$ Applying our argument to subsets of three codelines shows
that the remaining $7$ codelines are contained in a $PG(4,2).$ 
\par
Finally we use the fact some seven codelines $L_1,\dots ,L_7$ are contained in a $PG(4,2).$
Apply our argument to the following triples of codelines:
\begin{itemize}
\item
$\lbrace L_1,L_2,L_3\rbrace$ yielding $v(\lbrace L_1,L_2,L_3\rbrace)$ which we can choose to have nonzero entries in coordinates $1,2$ (and possibly $3$),
\item
$\lbrace L_2,L_3,L_4\rbrace$ where we choose
notation such that $2$ is in the support of $v(\lbrace L_2,L_3,L_4\rbrace),$ and
\item
$\lbrace L_3,L_4,L_5\rbrace$
\end{itemize}

This yields the contradiction $L_6=L_7.$ Proposition~\ref{pureprop} has been proved.

\section{The structure of the proof}
\label{structuresection}

Let $C$ be a $[[13,5,4]]$ quantum code, described by a set of lines $L_1,\dots ,L_{13}$
in the ambient space $U$ (a $PG(7,2)$). We know that the strength is $3.$ Let $e_1,\dots ,e_8$ be a basis of the
underlying vector space $V$ and choose $L_1=\langle e_1,e_2\rangle ,~L_2=\langle e_3,e_4\rangle .$ Consider the factor space $V/\langle e_1,e_2,e_3,e_4\rangle$ and the
corresponding $PG(3,2)$ which we call $\Pi.$ We work in $U$ and in the factor
space $\Pi .$ Because of strength $3$ each codeline $L_i,i>2$ defines a line in $\Pi.$

\begin{Definition}
\label{wdef}
Let $g$ be a line of $\Pi$ (a $PG(3,2)$). Define the {\bf weight} $w(g)$ of $g$
as $2$ less than the number of codelines contained in the preimage of $g.$
For points $P$ and planes $E$ of $\Pi$ define
$$w(P)=\sum_{P\in g}w(g),~w(E)=\sum_{g\subset E}w(g).$$
\end{Definition}

The geometric meaning of $w(P)$ and $w(E)$ is as follows: $w(P)+2$ is the number of codelines
which meet the preimage of $P$ (a $PG(4,2)$) nontrivially,
$w(E)+2$ is the number of codelines contained in the preimage of $E$ (a hyperplane $PG(6,2)$). 

\begin{Proposition}
\label{factorspaceprop}
We have $\sum_gw(g)=11$ where the sum is over all lines $g$ of $\Pi.$
For each line $h$ of $\Pi$ the number of lines of our multiset which intersect $h$ nontrivially is odd. 
\end{Proposition}
\begin{proof}
We think of the multiplicities $w(g)$ as defining a multiset, clearly of $11$ lines.
Let $h$ be a line of $\Pi.$ 
Its preimage under
the canonical mapping onto $\Pi$ is a secundum of the ambient space $U.$ 
The orthogonality condition of Definition~\ref{quantumcodedef} translates as follows
in geometric terms: for each secundum $S$ of $U$ the number of codelines meeting $S$
nontrivially is odd (see also~\cite{quantgeom}). Applying this to the preimage of
line $h$ yields our claim.
\end{proof}

We refer to the condition of Proposition~\ref{factorspaceprop} as the {\bf quantum condition.} Observe that in the quantum condition the sum is over all lines, including $h$ itself:
each of the $35$ lines of $\Pi$ gives a condition, and the sum is over all $g.$
\par
As $C$ is pure sets of strength $3$ play an important role.

\section{Sets of strength $3$}
\label{nmsetsection}

\begin{Definition}
\label{nmsetdef}
A set of objects in a projective space has {\bf strength} $3$ if any subset
of three of those objects are in general position.
An $(n,m)$-set is a set of strength $3$ consisting of $n$ lines and $m$ points.
\end{Definition}

\begin{Proposition}
\label{hyperplaneprop}
Assume $H$ is a hyperplane in $U$ containing precisely $n$ codelines.
Then $H$ meets the union of the codelines in an $(n,13-n)$-set whose points meet each hyperplane $S$ of $H$ in a cardinality whose parity is different from $n.$
\end{Proposition}
\begin{proof}
Each of the $13$ codelines either is contained in $H$ or it meets $H$ in a point. This
proves the first part. Let $S$ be a hyperplane of $H.$ Then $S$ is a secundum of $U$ and therefore meets an odd number of codelines. As $S$ does meet the $n$ codelines
contained in $H$ the second statement follows.
\end{proof}

In order to obtain bounds on $w(P),w(g),w(E)$ consider the
corresponding preimage spaces ($PG(4,2),PG(5,2),VPG(6,2),$ respectively)
with their $(n,m)$-sets formed by the intersection with codelines.

\begin{Lemma}
\label{V5lemma}
A $(2,m)$-set of strength $3$ in $PG(4,2)$ has $m\leq 4.$
All these sets are embedded in a uniquely determined $(2,4)$-set.
\end{Lemma}
\begin{proof}
The lines are without restriction
$L_1=\langle e_1,e_2\rangle , L_2=\langle e_3,e_4\rangle ,$
the points of the $(2,4)$-set of strength $3$ can be chosen as
$$e_5,~e_1+e_3+e_5,~e_2+e_4+e_5,~e_1+e_2+e_3+e_4+e_5.$$
\end{proof}

Of particular importance are
the hyperoval in $PG(2,4)$ and the $[7,3.5,4]_4$-codes. 

\begin{Lemma}
\label{HOlemma}
An $(n,0)$-set in $PG(5,2)$ has $n\leq 6.$
For each $n$ it is uniquely determined. They are all embedded in the
uniquely determined $(6,0)$-set, which we call the {\bf binary hyperoval.}
Consider $(n,m)$-sets in $V_6.$
If $n=6,$ then $m=0.$ If $n=5,$ then $m\leq 2.$
If $n=4,$ then $m\leq 4.$
\end{Lemma}
\begin{proof}
The first $3$ lines can be chosen as usual:
$$L_1=\langle e_1,e_2\rangle , L_2=\langle e_3,e_4\rangle , L_3=\langle e_5,e_6\rangle.$$
There are exactly $27$ points, the transversal 
points, each forming a $(3,1)$-set with $\lbrace L_1,L_2,L_3\rbrace.$
Then $L_4=\langle e_1+e_3+e_5,e_2+e_4+e_6\rangle$ is the essentially unique 
fourth line. There remain $6$ points each forming a $(4,1)$-set
together with $L_1,\dots ,L_4.$ 
These are exactly the six points on the remaining lines
$$L_5=\langle e_1+(e_3+e_4)+e_6,e_2+e_3+(e_5+e_6)\rangle ,
L_6=\langle e_1+e_4+(e_5+e_6),e_2+(e_3+e_4)+e_5\rangle$$
of the binary hyperoval. The uniqueness statement follows.
\end{proof}

We chose the term {\bf binary hyperoval} as the $(6,0)$-set in $PG(5,2)$ is
the binary image of the hyperoval in $PG(2,4).$ It is well known that the
hyperoval has the symmetric group $S_6$ as its group of automorphisms.
The automorphism group of the binary hyperoval has order $3\times 6!$ where the
additional factor $3$ stems from the multiplicative group of the field.
\par
As for the case of $(n,m)$-sets in $PG(6,2)$ we use earlier work in relation to
additive $[7,3.5,4]_4$-codes, see~\cite{addZ4,addJCTA08}.

\begin{Proposition}
\label{atmost7prop}
There is no $(7,0)$-set in $PG(5,2)$ and no
$(8,0)$-set in $PG(6,2).$ There are precisely three
non-equivalent $(7,0)$-sets in $PG(6,2).$ 
Exactly one of them defines a self-dual code with respect to the Euclidean
form (the dot product).
\end{Proposition}
\begin{proof}
A $(7,0)$-set in $V_6$ would define an additive $[7,4,4]_4$-code. In the same way an $(8,0)$-set in $V_7$
would lead to an $[8,4.5,4]_4$-code. Those codes do not exist.
\end{proof}

The classification of $(7,0)$-sets in $PG(6,2)$ has been carried
out independently several times, most recently in Danielsen-Parker~\cite{DP} and
Han-Kim~\cite{HanKim}. 

\begin{Proposition}
\label{71classiprop}
Consider the three $(7,0)$-sets in $PG(6,2).$ The number $c$ of points that complete them to a $(7,1)$-set is
$c=1,$ $c=2$ and $c=8,$ respectively.
The case of $8$ extension points occurs when the code generated
by the $(7,0)$-set is self-dual.
This $(7,0)$-set can be extended to a uniquely determined $(7,7)$-set and to a
$(7,6)$-set which is uniquely determined up to projectivity.
\end{Proposition}

\begin{proof}
This is a computer result.
The self-dual code is the one with $8$ extension points.
Here it is:

$$\left(\begin{array}{c|c|c|c|c|c|c}
L_1 & L_2 & L_3 & L_4 & L_5 & L_6 & L_7 \\
00 & 00 & 01 & 00 & 01 & 01 & 01 \\
01 & 00 & 00 & 01 & 00 & 01 & 01 \\
01 & 01 & 00 & 00 & 01 & 00 & 01 \\ 
00 & 00 & 10 & 10 & 10 & 00 & 10 \\
10 & 00 & 00 & 10 & 10 & 10 & 00 \\
00 & 10 & 00 & 00 & 10 & 10 & 10 \\ 
11 & 11 & 11 & 11 & 11 & 11 & 11 \\
 \end{array}\right)$$

The eight extension points are $P_0=(0:0:0:0:0:0:1)$ and the columns of

$$\left(\begin{array}{c|c|c|c|c|c|c}
0 & 0 & 1 & 0 & 1 & 1 & 1 \\
1 & 0 & 0 & 1 & 0 & 1 & 1 \\
1 & 1 & 0 & 0 & 1 & 0 & 1 \\ 
0 & 0 & 1 & 1 & 1 & 0 & 1 \\
1 & 0 & 0 & 1 & 1 & 1 & 0 \\
0 & 1 & 0 & 0 & 1 & 1 & 1 \\ 
1 & 1 & 1 & 1 & 1 & 1 & 1 \\
 \end{array}\right)$$

forming a set $E.$ 
The $(7,0)$-set has an automorphism group $G$ of order $42$ 
which fixes $P_0,$ preserves the hyperplane $H$ with equation $x_7=0$
and acts transitively on $E.$ 
Consider the cone with vertex $P_0$ consisting of the lines from $P_0$ 
to the points
$L_i\cap H.$ The third points on those lines make up $E.$
It follows that the uniquely determined $(7,7)$-set is defined by point set
$E$ and the essentially uniquely determined $(7,6)$-set is obtained by omitting
one point from $E.$
\end{proof}

\section{The weights in the factor space $PG(3,2)$}
Consider the weights $w(g)$ of lines in $\Pi=PG(3,2)$ and the
induced weights $w(P),w(E)$ on points and planes.

\begin{Lemma}
\label{wboundlemma}
For points, lines, planes of $\Pi$ we have 
$w(P)\leq 4, w(g)\leq 3, w(E)\leq 4.$
\end{Lemma}
\begin{proof}
The statement on points follows from Lemma~\ref{V5lemma}.
Proposition~\ref{atmost7prop} shows that $w(E)\leq 5.$
The hyperplane $H$ corresponding to a plane $E$ of weight $n$
yields an $(n+2,11-n)$-set. Assume $n=5.$ Then there is a $(7,6)$-set in $V_7.$ By Proposition~\ref{71classiprop}
the $7$
lines are uniquely determined as only the self-dual cyclic
example has more than $2$ extension points.
There is a uniquely determined $(7,6)$-set in $PG(6,2)$ (see Proposition~{71classiprop}),
but it does not satisfy the quantum condition
of Proposition~\ref{hyperplaneprop}. It follows $w(E)\leq 4.$
Assume now $w(g)=4.$ The quantum condition shows 
that it is contained in a plane of weight $5,$ contradiction.
\end{proof}

We can improve on Lemma~\ref{wboundlemma}:

\begin{Proposition}
\label{w(E)leq3prop}
$w(E)\leq 3$ for each plane $E$ of $\Pi.$
Each hyperplane $H$ of $U$ contains at most $5$ codelines.
The codelines define a quaternary $[13,4,8]$-code.
\end{Proposition}
\begin{proof}
All three statements of the proposition are equivalent.
Assume $w(E)=4.$ Assume at first $E$ contains a line $g$ such
that $w(g)=3.$
Then in the $PG(5,2)$ corresponding to $g$ we have the lines $L_1,\dots ,L_5$
corresponding to an oval in $PG(2,4)$ and $L_6=\langle e_1+e_4+e_5+e_6,e_7\rangle$ 
in the
hyperplane corresponding to $E.$ Those $6$ lines must be completable
to a $(6,7)$-system in $PG(6,2)$ which satisfies the quantum condition: 
each hyperplane of the $PG(6,2)$ must meet the set of $7$ 
extension points in odd cardinality. A computer search shows that this 
problem has no solution.
\par
Assume next $E$ contains a line $g$ of weight $2.$  
We have the usual lines $L_1,\dots ,L_4$ in $PG(5,2)$ and two more lines in the hyperplane 
which are not in the secundum. By Lemma~\ref{HOlemma} one of those lines can be chosen as
$L=\langle e_1+e_3+e_4+e_6,e_7\rangle .$ It remains to find the one remaining line
and the system of $7$ points in $PG(6,2)$ completing it to a $(6,7)$-system that satisfies the quantum condition.
A computer search shows that there is no solution. 
At this point we have shown the following:
\begin{itemize}
\item
Each hyperplane $H$ of $U$ which contains $6$
codelines is generated by each $4$ of its codelines.
\end{itemize}
This follows directly from the fact that for each plane
$E$ of weight $4$ of $\Pi$ we have $w(g)\leq 1$ for each
line $g\subset E.$ Observe that we could have started from any pair of codelines
instead of $L_1,L_2$ and considered the hyperplane corresponding to a plane of
weight $4$ in the factor space.
\par
A computer search showed that
there are exactly four families of $6$ lines in $PG(6,2)$ satisfying the following:
\begin{itemize}
\item
Any three of the lines are in general position.
\item
Any four of the lines generate the ambient space $PG(6,2).$
\end{itemize}

Here they are:

$$\left(\begin{array}{c|c|c|c|c|c}
L_1 & L_2 & L_3 & L_4 & L_5 & L_6 \\
10  & 00 & 00 & 10 & 00 & 01        \\
01  & 00 & 00 & 00 & 10 & 10    \\
00  & 10 & 00 & 10 & 10 & 10   \\
00  & 01 & 00 & 00 & 01 & 10   \\ \hline
00  & 00 & 10 & 10 & 01 & 11  \\
00  & 00 & 01 & 00 & 10 & 01   \\
00  & 00 & 00 & 01 & 01 & 01   \\
\end{array}\right)$$

$$\left(\begin{array}{c|c|c|c|c|c}
L_1 & L_2 & L_3 & L_4 & L_5 & L_6 \\
10  & 00 & 00 & 10 & 00 & 11        \\
01  & 00 & 00 & 00 & 10 & 10    \\
00  & 10 & 00 & 10 & 10 & 10   \\
00  & 01 & 00 & 00 & 01 & 11   \\ \hline
00  & 00 & 10 & 10 & 01 & 11  \\
00  & 00 & 01 & 00 & 10 & 11   \\
00  & 00 & 00 & 01 & 01 & 01   \\
\end{array}\right)$$

$$\left(\begin{array}{c|c|c|c|c|c}
L_1 & L_2 & L_3 & L_4 & L_5 & L_6 \\
10  & 00 & 00 & 10 & 00 & 11        \\
01  & 00 & 00 & 00 & 10 & 10    \\
00  & 10 & 00 & 10 & 01 & 10   \\
00  & 01 & 00 & 00 & 10 & 11   \\ \hline
00  & 00 & 10 & 10 & 10 & 11  \\
00  & 00 & 01 & 00 & 01 & 11   \\
00  & 00 & 00 & 01 & 01 & 01   \\
\end{array}\right)$$

$$\left(\begin{array}{c|c|c|c|c|c}
L_1 & L_2 & L_3 & L_4 & L_5 & L_6 \\
10  & 00 & 00 & 10 & 10 & 01        \\
01  & 00 & 00 & 00 & 11 & 10    \\
00  & 10 & 00 & 10 & 11 & 10   \\
00  & 01 & 00 & 00 & 01 & 10   \\ \hline
00  & 00 & 10 & 10 & 11 & 11  \\
00  & 00 & 01 & 00 & 11 & 01   \\
00  & 00 & 00 & 01 & 01 & 01   \\
\end{array}\right)$$

In  each of those cases another computer program shows that the corresponding family 
$F$ of codelines 
cannot be completed by a set $S$ of $7$ points in $H=PG(6,2)$ which together 
with the codelines form a $(6,7)$-set of strength $3$ and such that the quantum condition is satisfied.
\end{proof}

\section{Excluding a special configuration}
\label{penultsection}

In this section we show the following:

\begin{Proposition}
Any five codelines generate either the ambient space $U$ or a hyperplane.
\end{Proposition}

Assume this is not the case. If some five codelines were in a $PG(4,2)$ then
some hyperplane would contain six codelines, contradicting Proposition~\ref{w(E)leq3prop}.
Assume therefore some five codelines generate a secundum $S.$
In terms of the factor space $\Pi$ this means there is some
line $g_0$ of weight $3.$ As $w(E)\leq 3$ for each plane $E$ of
$\Pi$ this implies $w(g)=0$ for each 
line $g\not=g_0$ intersecting $g_0$ nontrivially. 
\par
The codelines in $S$ can be chosen as
$L_1,\dots ,L_5$ according to Lemma~\ref{HOlemma}. Let now $H\supset S$ be a hyperplane
and ${\cal M}=\lbrace M_0,\dots ,M_7\rbrace$ the points of intersection with
the eight remaining codelines. Then $M_i\notin S.$ Without restriction $M_0=e_7.$
Write $M_i=e_7+w_i.$ Then the following conditions must be satisfied:

\begin{enumerate}
\item
$w_i\notin L_1\cup\dots L_5$ for $i=1,\dots ,7.$
\item
$w_i+w_j\notin L_1\cup\dots L_5$ for $i\not=j.$
\item
Let $W$ be the $(7,8)$-matrix with the elements of ${\cal M}$ as columns. Then all codewords of the code generated by $W$ have even weights.
\end{enumerate}

Here the last condition represents the quantum condition: each hyperplane of $H$
meets ${\cal M}$ in even cardinality.
\par
A computer search showed that up to equivalence there are $12$ systems ${\cal M}$
satisfying the conditions above.
 
Here is the structure of the generator matrix that far:

$$\left(\begin{array}{c|c|c|c|c||c|l|c|c|c|c|c|c}
L_1 & L_2 & L_3 & L_4 & L_5 & L_6 & L_7 & L_8 & L_9 & L_{10} & L_{11} &L_{12}&L_{13}\\
10  & 00 & 00 & 10 & 10 & 00 & 1  &&&&&&                  \\
01  & 00 & 00 & 01 & 01 & 00 & 0  &&&&&&             \\
00  & 10 & 00 & 10 & 11 & 00 & 1  &&&&&&                \\
00  & 01 & 00 & 01 & 10 & 00 & 0  &&&&&&                     \\ \hline
00  & 00 & 10 & 10 & 01 & 00 & 0  &  &  &  &  &  &  \\
00  & 00 & 01 & 01 & 11 & 00 & 0  &  &  &  &  &  &   \\
00  & 00 & 00 & 00 & 00 & 10 & 10 & 10 & 10 & 10 & 10 & 10 & 10  \\
00  & 00 & 00 & 00 & 00 & 01 & 01 & 01 & 01 & 01 & 01 & 01 & 01  \\
\end{array}\right)$$ 

For each choice of ${\cal M}$ we need to determine the solutions of the problem in $PG(3,2)$ (the last four rows of the generator matrix).
Finally the generator matrix needs to be completed.
The computer showed that this completion is impossible.

\section{Completing the proof} 
\label{endsection}

Let $L_1,\dots ,L_5$ be codelines not generating the ambient space. They generate a hyperplane $H.$ Consider the corresponding $(5,8)$-set in $H.$ 
The lines define an additive $[5,3.5]_4$-code of strength $3.$
As its dual, a $[5,1.5,4]_4$-code, is uniquely determined
(corresponding to a set of $5$ lines in the Fano plane), the
same is true of the code itself. We can therefore choose
$$L_1=\langle e_1,e_2\rangle ,L_2=\langle e_3,e_4\rangle ,L_3=\langle e_5,e_6\rangle ,$$
$$L_4=\langle e_1+e_3+e_5,e_7\rangle ,L_5=\langle e_1+e_4+e_6,e_2+e_3+e_7\rangle .$$
No four of those are on a hyperplane.
How many points complete them to a $(5,1)$-set
of strength $3?$
There are $15$ points on the lines, $10\times 3/2=15$ in the intersection of the 
two spaces generated by two
lines and $10\times 6$ further points on spaces generated by two lines. 
This leaves space for $127-90=37$ extension points.
Within this set of $37$ points we have to find a subset ${\cal M}$ of eight points which satisfy the conditions

\begin{itemize}
\item
${\cal M}$ is a cap.
\item
Secants of ${\cal M}$ do not meet any of the lines $L_i.$
\item
Let $W$ be the $(7,8)$-matrix with the elements of ${\cal M}$ as columns. Then all codewords of the code generated by $W$ have even weights.
\end{itemize}

The general form of the generator matrix is

$$\left(\begin{array}{c|c|c|c|c||r|r|r|r|r|r|r|r}
L_1 & L_2 & L_3 & L_4 & L_5 & L_6 & L_7 & L_8 & L_9 & L_{10} & L_{11} &L_{12}&L_{13}\\
10  & 00 & 00 & 10 & 10 &  0 &   &&&&&&                  \\
01  & 00 & 00 & 00 & 01 &  0 &   &&&&&&             \\
00  & 10 & 00 & 10 & 01 &  0 &   &&&&&&                \\
00  & 01 & 00 & 00 & 10 &  0 &   &&&&&&                     \\ \hline
00  & 00 & 10 & 10 & 00 &  0 &   &  &  &  &  &  &  \\
00  & 00 & 01 & 00 & 10 &  0 &   &  &  &  &  &  &   \\
00  & 00 & 00 & 01 & 01 &  0 &    &    &    &    &    &    &    \\
00  & 00 & 00 & 00 & 00 & 01 & 01 & 01 & 01 & 01 & 01 & 01 & 01  \\
\end{array}\right)$$ 

A computer program did the following:

\begin{itemize}
\item
Determine the solutions ${\cal M}.$ 
\item
For each solution ${\cal M}$ determine the $8$ lines in $\Pi$ completing
the projections of the eight points of ${\cal M}$ such that the orthogonality
condition on the last four rows of the generator matrix are satisfied.
\item
Complete the generator matrix.
\end{itemize}

Observe that in the second step the projection to $\Pi$ may lead to repeated
points. This has to be taken into account when adapting the lines in $\Pi$
to the points of ${\cal M}.$ The computer search showed that there are no solutions.
This completes the proof of Theorem~\ref{maintheorem}.

\end{document}